\documentclass[11pt]{article}

\usepackage[margin=1 in]{geometry}
\usepackage{amsmath}
\usepackage{amssymb}
\usepackage{amstext}
\usepackage{amsthm}
\usepackage{natbib}
\usepackage{mathtools}
\usepackage{enumerate}
\usepackage{url}
\usepackage{dsfont}
\usepackage{float}
\usepackage{algorithm}
\usepackage{xspace}
\usepackage{subcaption} 
\usepackage{hyperref} 
\usepackage{algorithm}
\usepackage{algorithmic}
\usepackage{color}

\theoremstyle{definition}
\newtheorem{example}{Example}
\newtheorem{problem}{Problem}
\newtheorem{theorem}{Theorem}
\newtheorem{lemma}{Lemma}[section]
\newtheorem{corollary}[lemma]{Corollary}
\newtheorem{remark}[lemma]{Remark}
\newtheorem{prop}[lemma]{Proposition}
\newtheorem{claim}[lemma]{Claim}
\newtheorem{definition}[lemma]{Definition}

\newcommand{\problemName}{optimal ordering for optimal stopping}
\newcommand{\ProblemName}{Optimal ordering for optimal stopping}

\newcommand{\sacomment}[1]{}
\newcommand{\xzcomment}[1]{}
\newcommand{\jscomment}[1]{}

\newcommand{\blue}[1]{#1} 

\begin{document}
	\title{On optimal ordering in the optimal stopping problem}
	\author{Shipra Agrawal\\
	sa3305@columbia.edu\and
	Jay Sethuraman\\
	js1353@columbia.edu \and
	Xingyu Zhang\\
	xz2464@columbia.edu}
	\maketitle
	\begin{abstract}
In the classical optimal stopping problem, a player is given a sequence of random variables $X_1,\ldots, X_n$  with known distributions. After observing the realization of $X_i$, the player can either accept the observed reward from $X_i$ and stop, or reject the observed reward from $X_i$ and continue to observe the next variable $X_{i+1}$ in the sequence. Under any fixed  ordering of the random variables, an optimal stopping policy, one that maximizes the player's expected reward, is given by the solution of a simple dynamic program. In this paper, we investigate a relatively less studied question of selecting the  order in which the random variables should be observed so as to maximize the expected reward at the stopping time. Perhaps surprisingly, we demonstrate that this ordering problem is NP-hard even in a very restricted case where each random variable $X_i$ has a distribution with $3$-point support of form $\{0, m_i, 1\}$, and provide an FPTAS. We also provide a simple $O(n^2)$ algorithm for finding an optimal ordering in the case of $2$-point distributions. Further, we demonstrate the benefits of order selection, by proving a novel prophet inequality for $2$-point distributions that shows the optimal ordering can achieve an expected reward  within a factor of $1.25$ of the expected hindsight maximum; this is an improvement over the corresponding factor of $2$ for the worst-case ordering.

	\end{abstract}
	\section{Introduction}\label{Introduction}
Consider a player who can probe a sequence of $n$ independent random variables $X_1, \ldots, X_n$ with known distributions. After observing (the realized value of) $X_i$, the player needs to decide whether to stop and earn reward $X_i$, or reject the reward and probe the next variable $X_{i+1}$. The goal is to maximize the expected reward at the stopping time. This is an instance of the optimal stopping problem, which is a fundamental problem studied from many different aspects in mathematics, statistics, and computer science, and has found a wide variety of applications in sequential decision making and mechanism design. 

 When the order in which the random variables $X_1, \ldots, X_n$ are probed is fixed, the optimal stopping strategy can be found by solving a simple dynamic program. Under this strategy, at every step $i$, the player would compare the realized value of the current random variable $X_i$ to the expected reward (under the optimal strategy for the remaining subproblem) from the remaining variables $X_{i+1}, \ldots, X_n$, and stop if the former is greater than the latter.
The celebrated prophet inequalities compare the expected reward of the optimal stopping strategy to $E[\max(X_1, X_2, \ldots, X_n)]$, where the latter can be interpreted as the expected reward of a prophet who can foresee (or `prophesize') the values of all random variables in advance and therefore always stops at the random variable with maximum value. 

 A seminal result of Krengel and Sucheston~\cite{krengel1977semiamarts} upper bounds the ratio of the prophet's expected reward and that of an optimal stopping strategy by $2$,  for an arbitrary sequence of $n$ random variables. 
 Furthermore, they show that this bound is tight even for $n=2$.
 Surprisingly, Samuel-Cahn~\cite{samuel1984comparison} shows that we can achieve an approximation ratio of 2 using simpler stopping strategies (rather than an optimal stopping strategy); subsequent work by Chawla et al.~\cite{chawla2010multi} and Kleinberg and Weinberg~\cite{kleinberg2012matroid} identifies other stopping strategies that establish the same bound.   Prophet inequalities have since been used to design simple, sequential, posted-price mechanisms that guarantee a constant fraction of the social welfare or revenue achievable by any mechanism~ (see e.g. \cite{chawla2010multi, hajiaghayi2007automated}).

 In this paper, we focus on the relatively less studied question of optimizing the {\it order} in which the random variables should be probed. Specifically, besides choosing a stopping strategy, if the player is free to choose the order in which the random variables are probed,  then which ordering would maximize the expected reward at the stopping time? This question is motivated both by practical considerations and theoretical observations about the optimal stopping problem.

In practice, many decision-making settings allow the player such a choice of ordering. Consider for example, the problem of sequentially  interviewing candidates for a position, which is often presented as a canonical example of the optimal stopping problem. Assuming that the decision of hire/no-hire needs to be made for each candidate immediately after the interview, the interviewer wants to stop interviewing on reaching the candidate with the highest quality. In such a setting, the interviewer may have the liberty to decide the ordering in which the candidates are invited for an interview. Intuitively, given statistical information about the quality of each candidate (based on their resume for example), some orderings of the candidates can make  the decision problem easier and can ensure higher expected quality of the hired candidate. This study aims to formalize this intuition by studying the question of  finding an optimal ordering, as well as quantifying the gains to the expected reward from this additional degree of freedom.

Some insights into the impact of ordering can be obtained by considering prophet inequalities. For an arbitrary ordering, the prophet inequality cannot be improved beyond the approximation factor of $2$. 
On the other hand,  a prophet inequality with a much improved factor of  $\approx 1.342$ has been shown for i.i.d. random variables \cite{abolhassani2017beating,correa2017posted}. This gap between the identical and the non-identical case can be closed through better ordering. An example is the study in \cite{esfandiari2017prophet} that showed prophet inequality with a factor of $e/(e-1) \approx 1.6$ for random ordering (and later beaten by \cite{azar2018prophet})

Optimal orderings have potential to close this gap even further.
Consider the special case when all distributions have supports on two or fewer points (hence-forth referred to as the 2-point case). For arbitrary ordering, the prophet inequality cannot be improved beyond the factor of $2$ even in this special case. Here is a simple example  (similar to the example in \cite{esfandiari2017prophet}) that demonstrates this limitation. Let $X_1$ be a random variable that takes values $\{0,1\}$, with probabilities $\{1-\epsilon, \epsilon\}$ respectively, and $X_2$ takes value $\epsilon$ with probability $1$. Then, the prophet's expected reward is $E[\max(X_1, X_2)]=\epsilon + (1-\epsilon)\epsilon = 2\epsilon-\epsilon^2$. However, for the ordering $(X_2, X_1)$, the reward earned by the player has an expected value $\epsilon$ whether or not $X_2$ is accepted, yielding only a 2-approximation. On the other hand, for the ordering $(X_1,X_2)$, an expected reward that is the same as the prophet can be achieved by the following strategy: probe $X_1$, stop if $X_1=1$, otherwise probe $X_2$. Thus, in this example, the best ordering is better than the worst-case ordering by a factor of $2$ and the best ordering attains the same value as the prophet. Furthermore, in this example, choosing a random ordering results in an expected reward of $\approx \frac{3}{2} \epsilon$; and therefore, the best ordering is better than the random ordering by a factor of $4/3$.

These observations motivate our investigation into optimal orderings.
 \blue{The optimal ordering problem has been previously studied in mathematics and statistics literature in the 80's (e.g.,  \citet{gilat1987best}, \citet{hill1985selection}, and  \citet{hill1983prophet}). However the focus there has been on analytically characterizing the optimal order for some special structured cases (like Bernoulli and exponential distributions). Our focus is on understanding the computational complexity and devising tractable algorithms. One difficulty in such a study is that the nature of this problem changes significantly depending on the type of distributions considered. For example, when distributions are Bernoulli or exponential, the optimal ordering can be found analytically \citep{hill1985selection}, but, the problem remains nontrivial for uniform distributions, and as we show in this paper, even for distributions with very small support.}

\blue{Unlike the fixed ordering case, the problem of finding an optimal ordering for optimal  stopping  cannot be easily solved in polynomial time by dynamic programming. The ordering problem is an instance of the more general {\it stochastic dynamic programs}. Recently, \citet{fu2018ptas} provided a polynomial time approximation scheme (PTAS) for a class of stochastic dynamic programs under the assumption that all the distributions involved are supported on a constant number of points. The ordering problem studied here can be formulated as a problem in their class of stochastic dynamic program. In fact the optimal ordering problem is a special case of what they refer to as the committed Pandora's box problem, a variant of the  Pandora's box problem \cite{weitzman1979optimal}. 
}

\blue{In this work, we delve deeper into the optimal ordering problem when the distributions involved have a small support. What is the computational hardness of this problem? Can the problem can be solved to optimality under distributions with very small suport (2 or 3)? Does there exist an FPTAS\footnote{Fully Polynomial Time Approximation Scheme} ? Given that most interesting counter-examples and lower bounds for  prophet inequalities and impact of ordering (some of which were discussed above) have been shown for distributions with just 2-point support, the problem appears to be nontrivial even in these special cases.}  

\noindent \paragraph{\bf Our contributions.}\blue{We show that the optimal ordering problem is  NP-hard even under a special case of 3-point distributions where the highest and lowest points of the support are the same for all the distributions. This is  surprising, especially since our optimal ordering problem is a special case of the committed Pandora's box problem, which is a slight variant of the Pandora's box problem. And for the latter, an efficient optimal adaptive strategy is known for arbitrary distributions \cite{weitzman1979optimal}.  In fact, the hardness of the committed Pandora's box problem was not understood before our result, even for arbitrary distributions \cite{fu2018ptas}}.

\blue{Among positive results, we present an  FPTAS for a special case of 3-point distributions. We also devise an efficient polynomial time algorithm for finding an optimal ordering in the case of 2-point distributions.  Further, we show that in this case, under an optimal ordering the  prophet inequality holds with a significantly better approximation factor than that under an arbitrary ordering.} 

Our results are summarized as follows:
\begin{itemize}
\item {\bf NP hardness for 3-point distributions (Theorem \ref{th:hardness} in \S\ref{sec:hardness}).} Through a reduction from the subset product problem, we show that the problem of finding an optimal ordering is NP-hard even when  each random variable $X_i$ is restricted to be a 3-point distribution with support on $\{0,m_i,1\}$ for some $m_i \in (0,1)$, and \blue{$E[X_i|X_i > 0] = E[X_j|X_j>0]$ for all $i, j$}.  \sacomment{TODO: add comments about strong NP-hardness?}
\item {\bf Optimal ordering for 2-point distributions (Theorem \ref{th:Two-Point_ordering} in \S\ref{sec:two-point-ordering}).} We show that there exists a simple quadratic time algorithm for finding an optimal ordering in the 2-point case. 
\item {\bf New prophet inequality for 2-point distributions (Theorem \ref{th:prophet} in \S\ref{sec:prophet}).}   
 We prove that given any set of variables with 2-point distributions, under the {\it optimal ordering}, the prophet inequality holds  with a much improved factor of $1.25$ as compared to $2$ for an arbitrary ordering. And further, our prophet inequality is tight for 2-point distributions.  This illustrates the significance of the ability to choose an ordering. 
 \item {\bf FPTAS for 3-point distributions (Theorem \ref{th:FPTAS-different} in \S\ref{sec:fptas}).} We provide an FPTAS for the optimal ordering problem for the case when each random variable $X_i, i=1,\ldots, n$ has a three-point distribution with support on $\{a_i,m_i,1\}$ for some $a_i, m_i\in [0,1]$.
\end{itemize}

\subsection{Related Work}
Our work builds on the large body of work, starting from the early work on the classical prophet inequality, on finding an optimal ordering for optimal stopping, and on Pandora's box problem, to more recent work on the prophet secretary problem and variations. We briefly survey this literature and position our contributions in context.

\noindent \textbf{Pandora's box problem and stochastic dynamic programs.} The optimal stopping problem considered here is similar in spirit to a well-studied---but substantially easier--- problem called ``Pandora's box problem''~\cite{weitzman1979optimal}.  As in our model, there are $n$ random variables $X_1, X_2, \ldots, X_n$ with distributions is known to the decision maker. Also as in our model, the decision maker is free to  probe any random variable at any stage. Unlike in our model, however, in the Pandora's box problem, the decision maker is allowed to choose the value of {\em any} random variable that has been probed; and the feature that makes that problem non-trivial is that a random variable can be probed only at a (known) cost. Weitzman~\cite{weitzman1979optimal} constructs an ``index'' policy for that model and shows that such a policy is optimal; these indices are closely related to Gittins indices for the famous multi-armed bandit problem~\cite{gittins2011multi}. \blue{Our problem is more directly related to a variant called the {\it Committed Pandora's box problem} \cite{fu2018ptas}. Like in our setting, there the decision maker is committed to only choosing the value of the last probed random variable. Committed Pandora's box can be formulated as a problem in the  class of  stochastic dynamic program studied in \cite{fu2018ptas}, There the authors provide a general PTAS for any such problem when the distributions have a support on constant number of points.  Our hardness result for $3$-point distributions therefore provides a very strong computational hardness result for such stochastic dynamic programs.}


\noindent \textbf{Optimal Ordering for optimal stopping:} As mentioned earlier, the problem we consider was studied earlier by Hill~\cite{hill1983prophet} and Hill \& Hordijk~\cite{hill1985selection}. In \cite{hill1985selection}, the authors provide simple ordering rules for some families of random variables; this includes the case when every random variable $X_i$ is uniformly distributed between $0$ and some positive number $\alpha_i$, and some very specific cases of two-point distributions. Our result on optimal order for general two-point distributions requires significantly more work, and will reduce to their results in the specific cases studied there.
More importantly, these earlier papers~\cite{hill1983prophet, hill1985selection} also give examples for which simple rules of thumb---ordering based on mean or variance; stochastic ordering, assuming the variables are all stochastically ordered---do not work. Our NP-hardness result suggests that such heuristic rules are unlikely to be optimal. 

\noindent \textbf{Prophet Inequalities:} The work of Krengel and Sucheston~\cite{krengel1977semiamarts} generated a lot of interest in developing prophet inequalities and in obtaining simpler proofs; an important contribution is the work of Samuel-Cahn~\cite{samuel1984comparison}, who provided a simple rule that achieves the same approximation ratio. A good summary of this early work on classic prophet inequalities is the survey of Hill and Kertz~\cite{hill1992survey}. The work of Hajiaghayi et al.~\cite{hajiaghayi2007automated} and Chawla et al.~\cite{chawla2010multi} led to several novel applications of these ideas in designing mechanisms with provably good guarantees on social welfare and revenue. Since these papers, the community has generalized and extended the classical model to richer feasibility domains~\cite{kleinberg2012matroid, feldman2014simple, rubinstein2016beyond, rubinstein2017combinatorial, alaei2014bayesian, ehsani2018prophet, dutting2017prophet}. A good overview of this line of work is the survey paper of Lucier~\cite{lucier2017economic}.

\noindent \textbf{Prophet Secretary Problem:} In addition to extending the classical prophet inequality, researchers have also developed new models to better understand the role of the various assumptions of the basic model. As an example, one can ask if the approximation ratio can be improved if the order in which the random variables are drawn is itself chosen uniformly at random. This model, dubbed the {\em prophet secretary} problem, has been explored actively since its introduction by Esfandiari et al.~\cite{esfandiari2017prophet}, who showed that an improved approximation ratio of $e/(e-1) \approx 
1.58$ can be achieved. They achieved this by using a sequence of non-increasing thresholds on the random variables: in such a policy, we accept the $j$th random variable if its value exceeds the threshold $T_j$, regardless of the identity of the random variable being observed. Subsequent work by Correa et al.~\cite{correa2017posted} showed that the same result can be achieved by threshold policies in which the thresholds depend only on the random variable being observed, but independent of when it is observed. Somewhat surprisingly, in very recent work, Ehsani et al.~\cite{ehsani2018prophet} show that one can recover this bound using a {\em single} threshold combined with a carefully chosen (randomized) tie-breaking rule. That the bound of $e/(e-1) \approx 1.58$ is not optimal has been demonstrated in a series of papers, first by Azar et al.\cite{azar2018prophet} to approximately $1.576$ and then by Correa et al.~\cite{correa2019prophet} to approximately $1.503$. Interestingly, Correa et al.~\cite{correa2019prophet} also prove an upper bound of $(\sqrt3 + 1) /2 \; \approx 1.366$ on the approximation ratio achievable by any algorithm, even for the case of 2-point distributions.

\noindent \textbf{Order Selection:} \label{order-secretary intro} A few recent papers address the order selection version of the prophet inequality problem as well. The focus is not on finding an optimal ordering of the random variables, but on establishing (improved) bounds on the performance of an optimal ordering relative to that of a prophet. For example, Yan~\cite{yan2011mechanism} proves a bound of $e/(e-1)$ for this problem and later \cite{beyhaghi2018improved} improves the bound to $1.528$. Note that any bound on the prophet secretary problem automatically carries over to this case.  Interestingly,  our approximation ratio of $1.25$ for the case of 2-point distributions {\em cannot} be achieved for the prophet secretary version of the problem. This is implied by our simple example presented in the introduction.

\noindent \textbf{IID Instances:} We close by mentioning the recent developments on prophet inequalities for the case in which all the random variables are drawn from the same distribution. Note that order selection is irrelevant in this case, so that all three versions of the classic problem---worst case order, random order, and best case order---coincide. Hill and Kertz~\cite{hill1982comparisons} constructed a worst-case family of instances for each $n$ and later showed that the approximation ratio for these instances is at least $\approx 1.342$; this has been shown to be tight in a recent paper of Correa et al.~\cite{correa2017posted}. Because this is smaller than the worst-case bound for the prophet secretary model, we now know a separation between these two models. Interestingly, we do not know if the worst-case instances for the order selection problem are i.i.d. instances. 



	\section{Preliminaries}
\label{sec:prelims}

Let $X_1,\ldots, X_n$ be random variables with (known) distributions $D_1, \ldots, D_n$ respectively. An ordering is defined by a permutation $\sigma$ of indices $\{1,\ldots, n\}$.

The optimal stopping problem under a fixed ordering $\sigma$ is defined as follows.  
At each time step $t=1,\ldots, n$, a player first observes the value of $X_{\sigma(t)}$, generated independently from distribution $D_{\sigma(t)}$. Then, the player can either decide to stop and accept reward $X_{\sigma(t)}$; or reject $X_{\sigma(t)}$ and continue to round $t+1$. 
The stopping time $\tau$ is defined as the time step $t$ at which the player stops and accepts reward $X_{\sigma(t)}$. If the player reaches the end of the sequence without accepting any reward, then the game stops automatically, and the stopping time is defined as $\tau=n+1$ with reward $X_\tau:=0$ (as a consequence, we will never accept a negative realization from a random variable).
The goal is to maximize the expected reward $E[X_\tau]$ at the stopping time $\tau$, where the expectation is taken both over the stopping time (which may depend on the random instantiations of the past $X_{\sigma(t)}, t=1, \ldots, \tau-1$) and the distribution of $X_\tau$. 
We denote the  expected reward at optimal stopping time as $V_\sigma$. That is,  
\begin{equation}
\label{eq:Vsigma}
    V_\sigma := E[X_{\tau^*}]
\end{equation}
where $\tau^*$ is the stopping time given by the optimal stopping rule.

It is easy to see that, due to the optimal substructure property in this problem, the optimal stopping rule is defined by a simple Dynamic Program (DP). Specifically, let $V_\sigma(j)$ denote the optimal expected reward for the subsequence $X_{\sigma(j)}, \ldots, X_{\sigma(n)}$, so that $V_\sigma:=V_\sigma(1)$. Then, 
 for $j=1,\ldots, n$
\begin{eqnarray}
    V_\sigma(j) & = & E[\max(X_{\sigma(j)}, V_\sigma(j+1))], \text{ with}\\
    V_\sigma(n+1)& :=& 0.
    \end{eqnarray}
The following additional notation will be useful in proofs: for any sequence of random variables, say $S:=(X_1, \ldots, X_n)$, we define $V(X_1, \ldots, X_n)$  as
\begin{equation}\label{eq:notation1}
    V(X_1, \ldots, X_n) := E[\max(X_1, E[\max(X_2, \cdots, E[\max(X_{n-1}, E[\max(X_n, 0)])\cdots])])]
\end{equation}
 Sometimes, for simplicity, we will write $V(S)$ instead of $V(X_1,\ldots, X_n)$. Furthermore, if \newline $S = (X_1,\ldots, X_k)$ and $T = (X_{k+1},\ldots, X_m)$ are two disjoint sequences of random variables, then let $V(S, T):=V(X_1,\ldots, X_m)$. For any ordering $\sigma$, 
\begin{eqnarray} \label{eq:backward}
V_\sigma(j) =  V(X_{\sigma(j)}, \ldots, X_{\sigma(n)}) = V(X_{\sigma(j)}, \ldots, X_{\sigma(i)}, V_\sigma(i+1))
\end{eqnarray}  
for any $j<i<n$.

Given the above DP equations, the values $V_\sigma(1), \ldots, V_\sigma(n)$ can be calculated by backward induction.
Then, the optimal stopping policy is defined as follows: at any time step $t$, compare the realized value of the random variable $X_{\sigma(t)}$ to $V_\sigma(t+1)$; if this realized value is at least $V_\sigma(t+1)$, stop and accept the reward; otherwise, if $t< n$, continue to probe $X_{\sigma(t+1)}$. 

%

\begin{definition}[\ProblemName~problem]
\label{def:prob} For any given ordering $\sigma$ of $n$ random variables $X_1, \ldots, X_n$, let  $V_\sigma$ be the expected value at the optimal stopping time as defined in \eqref{eq:Vsigma}. We define the problem of {\problemName} as  the problem of choosing an ordering $\sigma$ that maximizes $V_\sigma$ i.e., the problem of finding 
\begin{equation} 
\sigma^* = \arg \max_\sigma V_\sigma.
\end{equation}
\end{definition}
Note that our definition of the optimal ordering problem is restricted to finding an optimal {\em static} ordering. In other words, the ordering of all the random variables is decided in advance based only on the distributions. In particular, the observed values of the random variables examined up to a particular stage are not used to dynamically change the ordering of the remaining variables. In fact, \textit{Hill~\cite[Theorem 3.11]{hill1983prophet} proves that there is always a static ordering of the variables that is optimal}.

\paragraph{k-point distributions.}
In this paper, we study computational hardness and present algorithmic methods for the optimal ordering problem. We investigate  the complexity of this problem for the case when the distributions have a finite $k$-point support. Specifically, we consider two-point and three-point distributions.

\begin{definition}[Two-point distributions] \label{def:two-point}A random variable $X_i$ with a two-point distribution is defined by three parameters $a_i, b_i, p_i$, and takes value 
  \[
X_i = \left\{\begin{array}{ll}
a_i, & \text{w.p. } 1-p_i,\\
b_i, & \text{w.p. } p_i.
\end{array}
\right.
\] 
Here, $a_i\le b_i$ are referred to as the left and the right end-point.
\end{definition}

\begin{definition}[Three-point distributions] 
\label{def:three-point}
A random variable $X_i$ with a three-point distribution is defined by five parameters $a_i, m_i, b_i, p_i, q_i$, and takes value 
  \[
X_i = \left\{\begin{array}{ll}
a_i, & \text{w.p. } 1-p_i - q_i, \\
m_i, & \text{w.p. } p_i, \\
b_i, & \text{w.p. } q_i.
\end{array}
\right.
\] 
Here, $a_i\le m_i\le b_i$, are referred to as the left end-point, the middle point, and the right end-point of the support, respectively.
\end{definition}



	\section{Hardness of optimal ordering}
\label{sec:hardness}

We show that the problem of finding an optimal ordering of the random variables is NP-hard even for the highly restricted special case in which each random variable $X_i$ is supported on exactly three points: $0$, $1$, and $m_i \in (0,1)$. 
Let $q_i := P(X_i = 1)$ and $p_i := P(X_i = m_i) $, so that $P(X_i = 0) = 1 - p_i - q_i$. We assume that $p_i > 0$, $q_i > 0$, and $p_i + q_i < 1$ for all $i$, so that each random variable can assume each value in its support with positive probability.
\begin{theorem}
\label{th:hardness}
The problem of optimal ordering for optimal stopping (refer to Definition \ref{def:prob}) is NP-hard for the case when for each $i=1,\ldots, n$, the random variable $X_i$ has a three-point distribution with support on $\{0,m_i,1\}$ for some $m_i \in (0,1)$. 
\end{theorem}
To prove the hardness result, we shall prove some useful properties on the structure of optimal orderings and optimal stopping rules for such random variables.

Fix any ordering $\sigma$ of the random variables $X_1, X_2, \ldots, X_n$. In this case,  the optimal stopping policy essentially partitions the $n$ variables into two categories: 
those which are accepted on being probed if and only if they realize their right endpoint $1$, i.e.,  the  (ordered) subset  ${S}^\sigma:=\{X_{\sigma(i)}, i\in [n]: V_{\sigma(i+1)}> m_{\sigma(i)}\}$; and  the remaining (ordered) subset ${T}^\sigma:= \{X_{\sigma(1)}, \ldots, X_{\sigma(n)}\} \setminus {S}^\sigma$. Note that since the last random variable is always accepted irrespective of its value, the last variable will always be in $T^\sigma$.

We claim that in {\em any} optimal ordering $\sigma$, the random variables in $S^\sigma$ appear before the random variables in $T^\sigma$. Further, the random variables in $T^\sigma$ must be ordered in weakly descending order of $E_i$, where
$$E_i := E[X_i | X_i > 0] \; = \; \frac{m_i p_i + q_i}{p_i+q_i}$$

\begin{claim}
\label{claim:ST}
Given random variables $X_1,\ldots, X_n$ that have three-point distributions with support $\{0,m_i, 1\}, $, and probabilities $\{1-p_i-q_i, p_i, q_i\}$ such that $m_i\in (0,1), p_i>0, q_i>0, p_i+q_i<1 $, for $i=1, \ldots, n$. Then, an ordering $\sigma$ is optimal if and only if (i) $S^\sigma$ precedes $T^\sigma$; (ii) the random variables in $S^\sigma$ are arranged arbitrarily; and (iii) the random variables in $T^\sigma$ appear in weakly decreasing order of $E_i$.
In particular, if $E_1 = E_2 = \ldots = E_n$, then the random variables in $T^\sigma$ can be arranged arbitrarily as well. 
\end{claim}
\begin{proof}
A complete  proof of this lemma is provided in Appendix \ref{app:hardness}. The proof uses an interchange argument, where we show that if in the optimal ordering $\sigma$, the ordering of any pair of variables $X_i, X_j$ violates the stated conditions, then they can be interchanged to increase $V_\sigma$ which would be a contradiction to the optimality of the ordering $\sigma$.   
\end{proof}

\begin{remark}
\label{rem:ST}
In fact, if the conditions $m_i\in (0,1), p_i>0, q_i>0, p_i+q_i<1 $ are not satisfied, e.g., if there are some variables with $m_i\in \{0,1\}$, or if the probability of one or more of the three support points is $0$, then proof of Claim \ref{claim:ST} can be modified to show that the conditions (i), (ii), (iii)  are still sufficient (though not necessary) for an ordering $\sigma$ to be optimal.
\end{remark}

\begin{definition}[Ordered partitions and Optimal partitioning problem]
\label{def:op}
A sequence $(S,T)$ of $n$ random variables  is an {\it ordered partition} if $S,T$ partitions the set of variables $X_1,\ldots, X_n$, $T$ is non-empty, the variables within $S$ are ordered arbitrarily, and the variables within $T$ are ordered in weakly descending order of $E_i$. 
The {\it optimal partitioning problem}  is defined as the problem of finding an ordered partition $(S,T)$ that maximizes $V(S,T)$ among all ordered partitions.  
\end{definition} 

\begin{corollary} A corollary of Claim \ref{claim:ST} is that in the three-point distribution case considered here, the optimal ordering problem is equivalent to the  optimal partitioning problem. 
\end{corollary}

We prove the NP-hardness of the optimal ordering problem by showing that finding an optimal partition is NP-hard. To that end, we consider the problem \textsc{Subset Product}, which is a multiplicative analog of the \textsc{Subset Sum} problem that is known to be NP-complete (see \cite{ng2010product}):

\begin{problem}\label{subset_product}
	\textsc{Subset Product}: Given integers $a_1,\ldots, a_n$ with each $a_i > 1$ and a positive integer $B$, is there a subset $T\subseteq N$ such that $\prod_{i\in T}a_i = B$?
\end{problem}

\begin{prop}
The  optimal partitioning problem (refer to Definition \ref{def:op}) is NP-hard when each of random variables $X_1,\ldots, X_n$ have three-point distributions with support $\{0,m_i, 1\}$, and probabilities $\{1-p_i-q_i, p_i, q_i\}$ such that $m_i\in (0,1), p_i>0, q_i>0, p_i+q_i<1 $, for $i=1, \ldots, n$.
\end{prop}

\begin{proof}
Given an instance of \textsc{Subset Product}, consider the following collection of random variables. Associated with each element $a_i$ is a random variable $X_i$ with distribution shown in Table~\ref{tab:xidist}. 

\begin{table}[h]
\centering 
\begin{tabular}{r ccc} 
\hline \\
Value & 0 & $\frac{B^2 - a_i}{B^2+1}$ & 1\\ \\
Probability & $\frac{1}{a_i^2}$ & $\frac{a_i-1}{a_i^2}$&  $\frac{a_i-1}{a_i}$\\ \\
\hline 
\end{tabular}
\caption{Distribution of $X_i$} 
\label{tab:xidist}
\end{table}

Notice that the $X_i$ has support $0$, and $1$, and $m_i = (B^2 - a_i) / (B^2 + 1)$. Notice also that as $a_i > 1$ for all $i$, $m_i \in (0,1)$, $0 <p_i < 1$, $0 <q_i < 1$ and $p_i + q_i < 1$ for all $i$. Finally, observe that
$$
E_i \; := \; E[X_i|X_i > 0] \; = \; \frac{(\frac{B^2-a_i}{B^2+1}) (\frac{a_i-1}{a_i^2})+ \frac{a_i-1}{a_i}}{1 - \frac{1}{a_i^2}}  \; = \; \frac{B^2}{B^2+1},$$
which is independent of $i$. Thus, in any ordered partition $(S,T)$ for this instance, the ordering within $S$ and the ordering within $T$ is irrelevant. 
	
To avoid cumbersome notation, we let $S$ and $T$ denote a partition of the indices $\{1,2, \ldots, n\}$. The expected reward $V(S,T)$ for an ordered  partition $(S, T)$ can be written as
\begin{eqnarray}
V(S,T) & = & 1 - \Pi_{i \in S} (1-q_i) + \bigg(\Pi_{i \in S} (1-q_i) \bigg) \bigg(1 - \Pi_{j \in T} (1-p_j-q_j) \bigg) \frac{B^2}{B^2+1}
\label{eq1}
\end{eqnarray}
An easy way to see why is to exploit the irrelevance of the relative ordering within $S$ and $T$: the decision maker earns 1 whenever any random variable in $S$ is observed to take on a value of 1; if none of the random variables in $S$ is accepted, the conditional expected value of any accepted random variable in $T$ is the same, and this possibility occurs unless every one of the random variables in $T$ is observed to be zero. Using the fact that $q_i = 1 - 1/a_i$ and $p_i = 1/a_i^2 - 1/a_i$, we can rewrite Eq.~(\ref{eq1}) as
$$V(S,T) \; = \;  1-\prod_{i\in S}\frac{1}{a_i} +\bigg(\prod_{i\in S}\frac{1}{a_i}\bigg)
		\bigg(1-\prod_{i\in T}\frac{1}{a_i^2}\bigg)\bigg(\frac{B^2}{B^2+1}\bigg).$$
	
Let $\gamma :=\prod_{i=1}^n a_i$, $\gamma_T:=\prod_{i\in T} a_i$, $\gamma_S:=\prod_{i\in S} a_i$. Then $\gamma_S = \gamma / \gamma_T$ and $V(S,T)$ can be written solely as a (one-dimensional) function of $\gamma_T$ as follows:
	
	$$V(S,T)=f(\gamma_T) := 1-\frac{\gamma_T}{\gamma} + \frac{\gamma_T}{\gamma}\bigg(1-\frac{1}{\gamma_T^2}\bigg)\frac{B^2}{B^2+1}$$ 
	
	Differentiating  $f(\cdot)$ with respect to $\gamma_T$ twice, we see that
	
	$$
	f'(\gamma_T) = -\frac{1}{\gamma} +\bigg(\frac{B^2}{B^2+1}\bigg)\bigg(\frac{1}{\gamma} + \frac{1}{\gamma\gamma_T^2}\bigg)
	$$
	
	and
	
	$$f''(\gamma_T) = \frac{-2B^2}{\gamma\gamma_T^3(B^2+1)} < 0$$
	
	Thus, $f(\gamma_T)$ is strictly concave in $\gamma_T$ and achieves its maximum when $f'(\cdot) = 0$, which occurs when $\gamma_T = B$. 

To complete the argument, we observe the following: given any instance of \textsc{Subset Product}, we construct the corresponding instance of our optimal partitioning problem and solve it to optimality. The optimal partition $(S,T)$ has $\gamma_T=B$ if and only if the given instance of \textsc{Subset Product} is a ``yes'' instance. Thus, the NP-completeness of \textsc{Subset Product} implies the NP-hardness of our optimal partitioning problem.
\end{proof}

\sacomment{TODO add some comments about strong NP-hardness?}

	\section{Ordering two-point distributions}
\label{sec:two-point}
In this section, we investigate optimal ordering and its significance for the case when all random variables $X_i, i=1,\ldots, n$ have two point distributions as defined in Definition \ref{def:two-point}. 
Recall that under two point distribution, a random variable $X_i$ can take two possible values $\{a_i, b_i\}$ (w.l.o.g. $a_i \le b_i$), with probability of the left end-point $a_i$ and the right end-point $b_i$ being $1-p_i$ and $p_i$, respectively. 

Our first result is the design of a simple and efficient algorithm for finding an optimal ordering in this case. Next, to illustrate the significance of being able to choose an ordering in this case, we prove a prophet inequality with an improved factor of $1.25$ for optimal ordering, as compared to the factor of $2$ for the worst-case ordering. We also show that this bound is tight, i.e., there exist two-point distributions under which even an optimal ordering cannot achieve an approximation ratio better than $1.25$ relative to the reward of the prophet. Thus, the stopping problem is non-trivial for two-point distributions.

We mention that the optimal order for several specific cases of two-point distributions were proven in Theorem 4.6 in \cite{hill1985selection}. Specifically, they prove the optimal order for the following families of random variables is by decreasing $\alpha_i$, where $\alpha_i$ is defined in one of the following ways:
    (1) $X_{\alpha_i}$ is a Bernoulli random variable with parameter ${\alpha_i} \in [0,1]$;
    (2) $X_{\alpha_i}$ is equally likely to be $\alpha_i$ or $-{\alpha_i}$;
    (3) $X_{\alpha_i}$ is $\alpha_i$ with probability $1/{\alpha_i}$ and $0$ otherwise; or  
    (4) $X_{\alpha_i}$ is equally likely to be $\alpha_i$ or ${\alpha_i} + 1$.

Here, we provide a much more general result by presenting an $O(n^2)$ algorithm for finding the optimal ordering given any collection of arbitrary two-point distributions.



\subsection{Optimal ordering algorithm}
\label{sec:two-point-ordering}
\newcommand{\propname}{Left Support Property\xspace}
\newcommand{\propabbrv}{LSP\xspace}

\begin{theorem}\label{th:Two-Point_ordering}
Given random variables $X_i, i=1,\ldots, n$ with arbitrary two-point distributions, there exists an algorithm to find the \problemName~in $O(n^2)$ time.
\end{theorem}
 
To derive the above result, we first investigate a simpler case when the left endpoint  $a_i=0$ for all $i$. In this case, the optimal ordering turns out to be very simple: the variables can be ordered simply in descending order of their right endpoints. Our key result is in Lemma \ref{lem:LEP} and Corollary \ref{cor:LEP} where we show a \propname (\propabbrv) of optimal ordering for any set of distributions. This property  enables us to extend the above  simple algorithm to obtain an optimal ordering algorithm for general two-point distributions. 


The following lemma characterizes the optimal ordering for the case when $a_i=0$ for all $i$.
\begin{lemma}\label{0lep-offlinemax}
    Given random variables $X_i, i=1,\ldots, n$ that are two-point distributions with $a_i=0, i = 1,\ldots, n$. Then, an optimal ordering can be obtained by ordering the variables in descending order of their right endpoints $b_i$. That is, an ordering $\sigma$ is optimal if $b_{\sigma(1)}\ge \cdots \ge b_{\sigma(n)}$. 
    Furthermore, under an additional condition that $a_i\ne b_i, 0<p_i <1$ for all $i$, an ordering $\sigma$ is optimal {\it only if}  $b_{\sigma(1)}\ge \cdots \ge b_{\sigma(n)}$.

\end{lemma}
\begin{proof}
 Given two-point distributions with $a_i=0$ and an ordering $\sigma$ with  $b_{\sigma(1)}\ge \cdots \ge b_{\sigma(n)}$, the optimal stopping policy (refer to Section \ref{sec:prelims}) reduces to the following: at any step $t$, check if $X_{\sigma(t)}=b_{\sigma(t)}$ (i.e., realizes its right endpoint); if yes, stop; otherwise continue. Since the variables are probed in descending order of right endpoints, the expected value $V_\sigma$ for this policy is simply equal to the maximum of right endpoints realized, or $0$. Since each $X_i$ can take either value $a_i=0$ or $b_i\ge 0$, this is same as $E[\max(X_1, \ldots, X_n)]$, that is, the expected hindsight maximum or the prophet's expected reward. Thus, trivially such an ordering $\sigma$ is optimal.

To prove the `only if' part of the lemma statement, suppose that for an ordering $\sigma$, there exists $j$ such that $b_{\sigma(j)} < b_{\sigma(j+1)}$. From the above discussion, there exists an ordering that achieves the expected hindsight maximum reward. Hence, it is sufficient to prove that with some positive probability, the reward at the stopping time under ordering $\sigma$ will be strictly smaller than the hindsight maximum. 

Consider the event that $X_{\sigma(j)}$ is probed and takes value $b_{\sigma(j)}$. This event will happen, e.g., if $X_{\sigma(i)}=0$ for all $i<j$ and $X_{\sigma(j)}=b_{\sigma(j)}$, which has positive probability since $0<p_i<1$ for all $i$. Under this event, there are two possible  scenarios for any stopping policy: either the stopping policy can accept $X_{\sigma(j)}=b_{\sigma(j)}$ and stop; or it can continue to probe the next variable. In both scenarios, we argue there is a positive probability that the hindsight reward is higher than the reward at stopping time. In the first scenario, this is the case if $X_{\sigma(j+1)}=b_{\sigma(j+1)}>b_{\sigma(j)}$ which can happen with probability $p_{\sigma(j+1)}>0$. In the second scenario, this is the case if $X_{\sigma(j+1)}=\cdots = X_{\sigma(n)}=0$, which  can happen with probability $\prod_{i\ge j+1} (1-p_{\sigma(i)})>0$.

	


\end{proof}

%

\begin{lemma}[\propname~(\propabbrv)]
\label{lem:LEP} Suppose $X_1$, $X_2$, $\ldots$, $X_n$ are random variables with bounded support  $X_i\in [a_i, b_i]$.
	Then there exists an optimal ordering $\sigma$ with the property that \mbox{$\Pr(X_{\sigma(i)} \le V_\sigma(i+1))>0$} for $i=1,\ldots, n-1$. That is, for every $X_i$, its distribution has non-zero support on the left of the value $V_\sigma(i+1)$. 
\end{lemma}
\begin{proof}
The proof is by construction. We show that from any optimal ordering we can obtain an  ordering that satisfies \propabbrv, without decreasing its value. W.l.o.g. let $\sigma=(1, \ldots, n)$ be an optimal ordering that does not satisfy \propabbrv. Let $X_i$ be the first r.v. among $X_1,\ldots, X_{n-1}$ which violates \propabbrv, i.e. $P(X_i \leq V_\sigma(i+1))=0$. Then, $X_i>V_\sigma(i+1)$ with probability $1$, and by definition $$V_\sigma(i)=E[\max(X_i, V_\sigma(i+1))] = E[X_i].$$ 
Now, consider an alternate ordering $$\sigma'=(1, \ldots, i-1, i+1, \ldots, n, i),$$
where variable $X_i$ is pushed to the end. 
We show that $\sigma'$ satisfies \propabbrv.

Observe that (refer to \eqref{eq:backward}),
$$V_{\sigma'}(i) = V(X_{i+1}, \ldots, X_n, X_i) =  E[\max(X_{i+1}, \cdots, E[\max(X_n,  E[\max(X_i, 0)])]\cdots)] \ge E[X_i]$$
Thus, 
\begin{equation}
    \label{eq:1}
    V_{\sigma'}(i) \ge V_\sigma(i) = E[X_i]
\end{equation} 
From above we can derive the following conclusions about the new ordering $\sigma'$:
\begin{itemize}
\item \propabbrv~ property is satisfied for indices  $i, \ldots n-1$ in $\sigma'$: Suppose for contradiction that for some $i\le j\le n-1$, \propabbrv property is violated, i.e., suppose that $X_{\sigma'(j)}>V_{\sigma'}(j+1)$ with probability $1$. Now, since
$$V_{\sigma'}(j+1)\ge V_{\sigma'}(n)  = E[\max(X_i, 0)]\ge E[X_i],$$
we have that $X_{\sigma'(j)}> E[X_i]$ with probability $1$. Since $\sigma'(j)\in \{i+1, \ldots, n\}$, this implies $V_\sigma(i+1)=V(X_{i+1}, \ldots, X_n)>E[X_i]$. For its expected value to be be below $V_\sigma(i+1)$,  $X_i$ must take value below $V_\sigma(i+1)$ with non-zero probability. This implies that \propabbrv property is  satisfied by $X_i$ in ordering $\sigma$,  which is a contradiction to the assumption we started with.
\item \propabbrv~ property is satisfied for  indices $1, \ldots, i-1$  in $\sigma'$:  Since  $\sigma(j)=\sigma'(j)$ for $1\le j\le i-1$,  and $V_{\sigma'}(i) \ge V_\sigma(i)$, we have that (refer to \eqref{eq:backward})
$$V_{\sigma'}(j) = V(X_j,\ldots, X_{i-1}, V_{\sigma'}(i))\ge V(X_j,\ldots, X_{i-1}, V_{\sigma}(i)) = V_\sigma(j).$$ 

This means for $1\le j\le i-1$, $$\Pr(X_{\sigma'(j)} \le V_{\sigma'}(j+1))\ge \Pr(X_\sigma(j) \le V_{\sigma}(j+1))>0,$$ 
where the last inequality followed from the assumption that $X_i$ is the first variable to violate OP in the ordering $\sigma$. Thus, the \propabbrv property is satisfied by $1\le j\le i-1$ in the new ordering.
\item $V_{\sigma'}\ge V_{\sigma}$: If $i=1$, $V_{\sigma'}(1)\ge V_\sigma(1) $ follows from \eqref{eq:1}, otherwise it follows from the observation in the previous bullet that $V_{\sigma'}(j)\ge V_{\sigma}(j)$ for $1\le j\le i-1$.
\end{itemize}
Thus, the new ordering $\sigma'$ satisfies \propabbrv~ property, and has  value greater than or equal to optimal ordering ($V_{\sigma'}\ge V_\sigma$). 
\end{proof}

	
	

A corollary of Lemma \ref{lem:LEP} is the following Left Endpoint Property (LEP) for  two-point distributions: 

\begin{corollary}[Left Endpoint Property (LEP) for two-point distributions] 
\label{cor:LEP}
Suppose $X_1$, $X_2$, $\ldots$, $X_n$ are random variables two-point distributions with support  $X_i\in \{a_i, b_i\}$.
	Then there exists an optimal ordering $\sigma$ with the property that $a_{\sigma(i)}\le  V_\sigma(i+1)$ for $i=1, \ldots, n-1$. 
	\end{corollary}

The left endpoint property  allows us to derive the following characterization of the optimal ordering in two-point distributions, which will significantly reduce the space of orderings to search over in order to find the optimal ordering.

\begin{prop} \label{Two-Point-Structure} 
Given $n$ variables with two-point distributions, define $n$ orderings as follows: for each $i=1,\ldots, n$, define  $\sigma^i$ as any ordering obtained by setting the last variable as $X_i$, and ordering the remaining variables in weakly descending order of their right endpoints. Then, at least one of these $n$ orderings is optimal.
\end{prop}

\begin{proof}
	By Corollary \ref{cor:LEP} there exists an optimal ordering satisfying LEP, w.l.o.g. assume it is $\sigma^*=(1, \ldots, n)$. 
	Let $\sigma$ be 
	an ordering such that $\sigma(n)=n$, and $b_{\sigma(1)}\geq \ldots \geq b_{\sigma(n-1)}$, i.e., $\sigma$ is one of the $n$ orderings defined in the proposition statement. Then, we show that  $V_\sigma\ge V(X_1,\ldots, X_n)$.
	
	W.l.o.g., we assume that $b_i > E[\max{X_n, 0}]$ for all $i$, otherwise, $X_i$ will always be rejected in both $\sigma$ and $\sigma^*$. We will show that $V_{\sigma} = V(X_{\sigma(1)},\ldots, X_{\sigma(n-1)}, X_n)\geq V(X_1,\ldots, X_{n-1}, X_n) = V_{\sigma^*}$.
	
	For $i=1,\ldots, n-1$, define:
	
	\[
	X_i' = \left\{\begin{array}{ll}
	E[\max(X_n, 0)], & \text{w.p. } 1-p_i,\\
	b_i, & \text{w.p. } p_i
	\end{array}
	\right\}
	\] 
	and
		\[
	X_i'' = \left\{\begin{array}{ll}
	0, & \text{w.p. } 1-p_i,\\
	b_i - E[\max(X_n,0)], & \text{w.p. } p_i
	\end{array}\right\}
	\] 
	
We claim the following sequence of relations:

\begin{align*}
    V(X_{\sigma(1)},\ldots, X_{\sigma(n-1)}, X_n) & \geq V(X'_{\sigma(1)},\ldots, X'_{\sigma(n-1)}, X_n) \tag{1} \\
    &=V(X'_{\sigma(1)},\ldots, X'_{\sigma(n-1)}) \tag{2} \\
    &=V(X''_{\sigma(1)},\ldots, X''_{\sigma(n-1)}) + E[\max(X_n, 0)] \tag{3}\\
    &\geq V(X''_{1},\ldots, X''_{n-1}) + E[\max(X_n, 0)] \tag{4} \\
    &= V(X'_{1},\ldots, X'_{n-1})\tag{5}\\
    &= V(X'_{1},\ldots, X'_{n-1}, X_n)\tag{6}\\
    &= V(X_1,\ldots, X_n)\tag{7}
\end{align*}

We now justify each of the relations:

For $(1)$: First notice that $V_{\sigma}(i+1)\geq E[\max(X_n,0)]$ for all $i<n$. In case, $a_{\sigma(i)} \leq E[\max(X_n,0)] \le V_{\sigma}(i+1)$, then $X_{\sigma(i)}$ would be rejected if its left endpoint is realized, therefore, increasing its left endpoint to $E[\max(X_n,0)] $ does not change the overall expected value. On the other hand, if $a_{\sigma(i)} > E[\max(X_n,0)]$, then transforming by decreasing the left endpoint from $a_{\sigma(i)}$ to $E[X_n]$ can only decrease (or not change) the overall expected return.

For $(2)$ and $(6)$: If for any of the two orderings, the second last variable ($X'_{\sigma(n-1)}$ or $X'_{n-1}$) is probed and its left endpoint ($E[\max(X_n,0)]$) is realized, then accepting it would give  $E[\max(X_n,0)]$ reward, while rejecting it and continuing would also give  $E[\max(X_n, 0)]$. Thus, removing $X_n$ does not affect the overall expected reward.

For the $(3)$ and $(5)$: $X_i'\ge 0$ and $X_i''=X_i'-E[\max(X_n,0)]\ge 0$ due to the assumption made (w.l.o.g.) that $b_i \geq E[\max(X_n, 0)]$; we can therefore use an observation that for any sequence of variables $(Y_1,\ldots Y_k)$  and constant $c$ if $Y_i+c\ge 0$, then $V(Y_1+c,\ldots Y_k+c) = V(Y_1,\ldots Y_k) + c$. This is formally proven in Lemma \ref{scaling} in the appendix.

For $(4)$: Since $X_i''$ are two-point distributions with $0$ left endpoint, it follows from Lemma \ref{0lep-offlinemax} since $\sigma$ is an optimal ordering.

For $(7)$: This is due to the fact that $\sigma^*$ is an LEP ordering. Therefore, $a_i \leq V_{\sigma^*}(i+1)$ for all $i<n$, so that if $X_i$ realizes its left endpoint, it will be rejected anyway. Hence, changing its left endpoint to $E[\max(0, X_n)]\le V_{\sigma^*}(i+1)$ will not change its overall expected value. 

\end{proof}
\begin{proof}[Proof of Theorem~\ref{th:Two-Point_ordering}]
Proposition \ref{Two-Point-Structure} narrows down the space of orderings to be searched over to just $n$ orderings $\sigma^i, i=1,\ldots, n$, where the ordering $\sigma^i$ was defined in Proposition \ref{Two-Point-Structure}. 
It will take $O(n)$ time to compute the expected reward $V_{\sigma^i}$ of each of these $n$ orderings, and therefore takes $O(n^2)$ time to compute the expected reward for all $n$ orderings and find the best ordering.
\end{proof}


\subsection{Prophet inequality for optimal ordering}
\label{sec:prophet}

In the previous subsection, we presented an  an algorithm for finding the optimal ordering for two-point distributions. In this section, we prove a prophet inequality bounding the ratio of the prophet's reward and the expected reward under best ordering by $1.25$. 
In comparison, as noted in \ref{order-secretary intro}, there exist (two-point) distributions for which this factor is $2$ under worst-case ordering \cite{samuel1984comparison}. This illustrates the benefits of ordering. 

\begin{theorem}[Prophet inequality for two-point distributions]
\label{th:prophet}
Given any set of $n$ random variables $X_1, \ldots, X_n$ with two-point distributions, the prophet's expected reward is within $1.25$ factor of the expected reward at stopping time under optimal ordering, i.e., 
\begin{equation}
\label{eq:prophet}
    \frac{E[\max(X_1,\ldots, X_n)]}{(\max_\sigma V_\sigma)} \le 1.25
\end{equation}
where $V_\sigma$ is as defined in \eqref{eq:Vsigma}.
\end{theorem}







\begin{proof}
Let $X_1,\ldots, X_n$ be a set of random variables with two-point distributions, where $X_i$ takes values $\{a_i, b_i\}$, $a_i\le b_i$, with probabilities $1-p_i$ and $p_i$, respectively. Let $X^*:=X_{i^*}$ be the random variable with largest left endpoint, i.e., $i^*=\arg \max_i a_i$, with support points and probability denoted as $\{a^*, b^*, p^*\}=\{a_{i^*}, b_{i^*}, p_{i^*}\}$.
Let $U:=\{i: b_i \ge b^*\}\backslash i^*$, $W:=\{i: b^* > b_i \ge a^*\}$. Let the variables in $U, W$ be ordered in weakly descending order of their right endpoints $b_i$. 

Note that by definition of $U, W$ for any $i\notin U\cup W \cup \{i^*\}$, the right endpoint $b_i$ must be strictly smaller than the left endpoint $a^*$ of $X^*$. Thus, such an $X_i$ will always take a value smaller than $\max(X_1,X_2, \ldots, X_n)$ so that $E[\max(X_1,\ldots, X_n)]=E[\max_{i\in U \cup W\cup \{i^*\}} X_i]$. Further, for any ordering $\sigma$, let $\bar \sigma$ be the ordering restricted to the subset of variables in  $U \cup W\cup \{i^*\}$; then $V_\sigma \ge V(X_{\bar \sigma})$, where $X_{\bar{\sigma}}$ denotes the variables in $U \cup W\cup \{i^*\}$ ordered according to $\bar \sigma$. Thus, ignoring the remaining variables $\{i\notin U \cup W\cup \{i^*\}\}$ can only hurt the prophet inequality, since the numerator would remain the same while the denominator can only decrease.  We therefore ignore such variables in the remaining discussion, and assume w.l.o.g. that $U\cup W \cup \{i^*\} = \{1,\ldots, n\}$.

Next, we prove the prophet inequality \eqref{eq:prophet} by showing that the expected reward  at stopping time under one of the following two orderings is at least $4/5$ of the expected hindsight maximum. 

\begin{enumerate}
	\item Ordering $\sigma^1=(U, W, i^*)$. In this ordering, we first have the variables with right endpoint larger than $b^*$, in (weakly) decreasing right endpoint order. Then,  the variables with right endpoint smaller than $b^*$ (but larger than $a^*$), in (weakly) decreasing right endpoint order. And, finally we have $i^*$.
	\item Ordering  $\sigma^2=(U, i^*, W)$. That is, essentially the variables are ordered in (weakly) decreasing right endpoint order.
\end{enumerate}
An easy scenario is when at least one of the variables in $i \in U$ realizes its right endpoint 
$b_i$. In this scenario, under both the above orderings $\sigma \in \{\sigma^1, \sigma^2\}$, the optimal stopping policy  will  achieve the same reward as the prophet. 
To see this, first note that since 
$$\min_{i\in U} b_i \ge b^* \ge \max_{i\in \{V\cup i^*\}} b_i \ge V_\sigma(|U|+1),$$ the variables in $W\cup \{i^*\}$ 
will never be probed in this scenario. Among those in $U$, the optimal stopping policy will reject all variables in $U$ whose left endpoint is realized, and accept the first variable whose right endpoint is realized. This is because, $a^*\ge a_i$ for all $i\in U$, so that $  V_\sigma(|U|+1)\ge a^*\ge a_i$, and because the variables in $U$ are ordered in decreasing right endpoint order.   
Therefore, the optimal stopping policy will achieve a reward of  $\max_{i\in U: X_i=b_i} b_i$. This will also be the prophet's reward, as the maximum value can never be achieved by some left endpoint of $U$ (because $a_i\le a^*, \forall i\in U$) and none of the variables in $W\cup \{i^*\}$ can achieve a higher value than $\max_{i\in U: X_i=b_i} b_i$. 

In the remaining proof we consider the alternate scenario that every variable in $i \in U$ realizes its left  endpoint $a_i$. Furthermore, we assume that $W$ is non-empty, otherwise we will just accept $X^*$. Since $a_i\le a^*$ for all $i$, in this scenario, the prophet's reward is $\max_{i\in W\cup \{i^*\}} X_i$; and since $V_\sigma(|U|+1)\ge a^*$, the stopping policy for both orderings $\sigma\in \{\sigma^1, \sigma^2\}$ will reject all the variables in $U$. Thus, the variables in $U$ do not effect the reward of either the prophet or at the stopping time, and we ignore them for the remaining discussion and focus on the reward obtained from the variables in $W\cup \{i^*\}$.

Under the two orderings $\sigma^1, \sigma^2$, the random variables are ordered as $(W, i^*)$ and $(i^*, W)$, respectively, so that $V_{\sigma^1}= V(X_W, X^*)$ and $V_{\sigma^2}= V(X^*, X_W)$, where $X_W$ denotes the sequence of random variables in $W$. 

Now, define quantities 
$$b_w:=
E[\max_{i\in W} X_i| \exists i\in W: X_i\ge b_i], \text{ and } p_w:=\Pr(\exists i\in W: X_i\ge b_i)=1-\prod_{i\in W}  (1-p_i).$$
Then,  we can make the following observations about the prophet's expected reward and the maximum expected reward under the orderings $\sigma^1, \sigma^2$.

Firstly, since $b^*\ge b_w\ge a^*$, the prophet's maximum reward is given by
\begin{equation}
\label{eq:max}
   MAX:= E[\max(X_W, X^*)] = p^* b^* + (1-p^*)p_w b_w + (1-p^*) (1-p_w) a^*
\end{equation}
For the ordering $\sigma^1$,  since $V(X^*) = \mu^*:=p^*b^* + (1-p^*) a^* $ and variables in $W$ are ordered in decreasing right endpoint, the optimal stopping policy will either accept some $X_i$ in $W$ with the largest realized right endpoint $b_i$ as long as it is more than $\mu^*$, 
or wait till $i^*$ to get expected reward $\mu^*$, so that the expected reward
$$V_{\sigma^1} = V(X_W, \mu^*) \ge  E[\max(\max_{i\in W:X_i=b_i} b_i, \mu^*)] \ge p_w b_w + (1-p_w) \mu^*$$

For the ordering $\sigma^2$,  since the variables in $W$ are ordered in decreasing right endpoint, when probing $X_W$, the optimal stopping policy will  accept the first variable that realizes its right endpoint, so that $ V(X_W) \ge  p_w b_w$ 
and 
$$ V_{\sigma^2} = V(X^*, X_W) \ge  E[\max(X^*, p_wb_b)] \ge \max(\mu^*, p^*b^* + (1-p^*) p_w b_w )$$




Combining, we have
\begin{equation}
\label{eq:T123}
    \max \left(V_{\sigma^1}, V_{\sigma^2}\right) \ge \max\left(\underbrace{p_w b_w + (1-p_w) \mu^*}_{T1}\ ,\  \underbrace{p^*b^* + (1-p^*) p_w b_w}_{T2}\ ,\  \underbrace{\mu^*}_{T3}\right)
\end{equation}

Next, we complete the proof of the prophet inequality by showing that $\max(T1, T2, T3) \ge 0.8 MAX$, where $MAX$ was defined in \eqref{eq:max} as the prophet's expected reward. The proof of this statement is largely algebraic and is provided in Appendix \ref{app:prophet}. 
\end{proof}

\paragraph{Tightness.}
Here is an example where the expected reward at stopping time under best ordering is arbitrarily close to $0.8 * E[\max_i X_i]$. 

\begin{example}
Let $X_1$ have support $\{0.5, \frac{1}{2\epsilon}\}$ with probabilities $\{1-\epsilon, \epsilon\}$, and $X_2$ have support $\{0, 1\}$ with probabilities $\{0.5, 0.5\}$. Then

\begin{eqnarray*}
E[\max_i X_i] &=& \epsilon \frac{1}{2\epsilon} + (1-\epsilon)(1/2)*1 + (1-\epsilon)(1/2)*0.5\\
&=&\frac{5}{4} - \frac{3}{4}\epsilon
\end{eqnarray*}

The two possible orderings $(X_1, X_2)$ and $(X_2, X_1)$ are equally good here:

$$V(X_1, X_2) = \epsilon \frac{1}{2\epsilon} + (1-\epsilon)0.5\\
=1-0.5\epsilon$$

$$V(X_2, X_1) = 0.5 * 1 + 0.5*(\epsilon \frac{1}{2\epsilon} + (1-\epsilon)0.5) = 1 -0.5\epsilon$$

So in this example, the best ordering returns $0.8$ of the offline maximum.
\end{example}

	\section{Ordering three point distributions: FPTAS}
\label{sec:three-point}
\label{sec:fptas}
\newcommand{\OPT}{\text{OPT}}
\newcommand{\ALG}{\text{ALG}}
\sacomment{We should add some discussion here on how our FPTAS is better than their PTAS. May be instead of using $poly(n,1/\epsilon)$ we should use exact expression, so that we can compare to theirs. Update: I have added text for this. I also did a pass and moved the proof to the appendix.}
Previously, in Section \ref{sec:hardness}, we proved that ordering $n$ random variables with three-point distributions is  NP-hard even when each random variable $X_i$ has three point distribution with support on $\{0,m_i,1\}$. \blue{In this section, we provide an FPTAS  for a slight generalization of this special case where the support is on three points $\{a_i, m_i, 1\}$ where $a_i < m_i <1$. Note that in \cite{fu2018ptas}, the authors provide a PTAS for this problem with distributions that have support on any constant number points. The runtime complexity of the PTAS proposed there is $O(n^{2^\epsilon})$ in general. However, (to the best of our understanding) when applied to our problem, its complexity reduces to $O\left(\left(\frac{n}{\epsilon}\right)^{\epsilon^{-3}}\right)$. On other hand, for any $\epsilon\in(0,1)$, our FPTAS runs in time $O(\frac{n^5}{\epsilon^2})$ (although it only applies to the special case of $3$-point support).}

\begin{theorem}[FPTAS for three-point distributions with same  right endpoint]
\label{th:FPTAS-different}
Given a set of $n$ random variables $X_1,\ldots, X_n$, where each random variable $X_i, i=1,\ldots, n$ has three-point distribution with support on $\{a_i,m_i,1\}$ for some $m_i\in [0,1]$ and $a_i<m_i$. 
 Then, there exists an algorithm that runs in time $O(\frac{n^5}{\epsilon^2})$ to find an ordering $\sigma$ such that $\ALG=V_\sigma \ge (1-\epsilon)\OPT$. Here, $\OPT:=V_{\sigma^*}$ denotes the optimal expected reward at stopping time under an optimal ordering $\sigma^*$.
\end{theorem} 
 
 We first demonstrate an FPTAS for an the special case where both left and right end points are the same for all $i$. Later we extend it to an FPTAS for the case when only the right end points are the same to prove Theorem \ref{th:FPTAS-different}.

 
\paragraph{Algorithm for same left and right end points.} We are given a set of random variables $X_1, \ldots, X_n$ with three-point distributions (refer to Definition \ref{def:three-point}) where for each $i$, $X_i$  takes values  $\{0, m_i, 1\}$  with probabilities $1-p_i-q_i$, $p_i$ and $q_i$, respectively.  The algorithm for finding an optimal ordering is based on the characterization of optimal orderings in this special case provided by Claim \ref{claim:ST} in Section \ref{sec:hardness} (also see Remark \ref{rem:ST}). 

Recall that in this case, there exists an optimal ordering $\sigma$
that (under the optimal stopping policy) partitions the $n$ variables into an ordered partition $(S^\sigma, T^\sigma)$ (refer to Definition \ref{def:op}); and therefore the optimal ordering problem can be solved by finding an optimal ordered partition. 
Specifically, define the collection $\mathcal{L}$ of ordered partitions  as the collection of sequences $(S,T)$ of $n$ random variables that can be formed by partitioning the set of variables $\{X_1, \ldots X_n\}$ into two sets $S$ and $T$ with $T\ne \phi$, and then ordering the variables within $S$ and within $T$ in weakly descending order of $E_i$. Then, from Claim \ref{claim:ST}, we have that
 \begin{equation}
    \label{eq:paritionProblem}
\OPT = \max_{(S, T)\in \mathcal{L}} V(S, T)
\end{equation}
Using this observation, Algorithm \ref{alg:fptas} is designed to solve the problem in  \eqref{eq:paritionProblem} of finding an optimal ordered partition. 
The idea behind the FPTAS is to discretize the interval $[0,1]$ using a multiplicative grid with parameter $1-\frac{\epsilon}{2n}$,  so that it needs to search over only $poly(n,1/\epsilon)$ partitions. 
A detailed description of the steps involved is provided below (Algorithm \ref{alg:fptas}). Here, given a sequence $A$ of  random variables, $\{X_i, A\}$  denotes the sequence of random variables formed by concatenating a variable $X_i$ to the beginning of sequence $A$.
\newcommand{\MAX}{\text{MAX}}

\begin{algorithm}[h]
\begin{algorithmic}
\caption{FPTAS for finding the optimal ordering  through optimal partitioning
\label{alg:fptas}
}
\STATE {\bf Input:} Ordered sequence of variables $X_1, \ldots, X_n$ such that $E_1 \le \cdots \le E_n$, parameters $\MAX$, $\epsilon$.
\STATE {\bf Initialize:}  $\mathcal{L}^0=\{(\phi, \phi)\}, \mathcal{L}^1 = \cdots = \mathcal{L}^n=\phi$; 
\FORALL{$k=1, \ldots, n$}
    \FORALL{$(S,T)\in \mathcal{L}^{k-1}$}
        \STATE Add two partitions $(\{X_k, S\}, T)$ and  $(S, \{X_k, T\})$ to $\mathcal{L}^k$.
    \ENDFOR
    \STATE Call Algorithm \ref{alg:trim} to reduce the number of partitions in $\mathcal{L}^k$ by setting $\mathcal{L}^k\leftarrow \text{TRIM}(\mathcal{L}^k, \epsilon, \MAX)$.
\ENDFOR
\STATE Return $\mathcal{L}^n$.
\end{algorithmic}
\end{algorithm}

\begin{algorithm}[h]
\begin{algorithmic}
    \caption{TRIM($\mathcal{L}, \epsilon, \MAX$)
    \label{alg:trim}
    }
    \STATE {\bf Initialize:} $\rho :=\left(1-\frac{\epsilon }{2n}\right)$, $max:=\max_{(S,T)\in \mathcal{L}} V(T)$,  
    and 
    $J:=\max\{j: \rho^j max\ge  \frac{\epsilon}{2n}\MAX\}.$ 
    \STATE Divide the partitions in $\mathcal{L}$ into $J+1$ buckets as 
    \begin{center}$\mathcal{B}_j := \{(S,T): \rho^{j} 
    max < V(T) \le \rho^{j-1} max
    \}, \text{ for } j=1,\ldots, J$
    \end{center}
    $$\mathcal{B}_0 := \{(S, T): T=\phi\}$$
    \STATE  Set  $(S^j, T^j) := \arg \max_{(S,T)\in \mathcal{B}_j} V(S)$, for $j=0,1,\ldots, J$.
     \STATE Return $\mathcal{L} := \{(S^j, T^j)\}_{j=0}^J$. 
\end{algorithmic}
\end{algorithm}

We prove the following theorem regarding Algorithm \ref{alg:fptas}.
The proof is in the appendix.

\begin{theorem}[FPTAS for three-point distributions with same left and right endpoint]
\label{th:FPTAS-same}
Given a set of $n$ random variables $X_1,\ldots, X_n$, where each random variable $X_i, i=1,\ldots, n$ has three-point distribution with support on $\{0,m_i,1\}$ for some $m_i\in [0,1]$. 
 Then, Algorithm \ref{alg:fptas} runs in time $O(\frac{n^4}{\epsilon^2})$ and finds an ordering $\sigma$ such that $\ALG=V_\sigma \ge (1-\epsilon)\OPT$. Here, $\OPT:=V_{\sigma^*}$ denotes the optimal expected reward at stopping time under an optimal ordering $\sigma^*$.
\end{theorem}

Now we are ready to prove Theorem~\ref{th:FPTAS-different} using an extension of Algorithm \ref{alg:fptas}.
\paragraph{\bf Proof of Theorem \ref{th:FPTAS-different}}
Let $q_i:=P(X_i=1)$, $p_i:=P(X_i=m_i)$ and so $P(X_i=0)=1-p_i-q_i$. From Lemma \ref{lem:LEP}, we know that there exists an optimal ordering $\sigma$ with the property that $Pr(X_{\sigma(i)}\leq V_\sigma(i+1))>0$ for $i=1,\ldots,n-1$. This implies that $a_{\sigma(i)}\leq V_\sigma(i+1)$ and, as a consequence of the DP thresholds, if $a_{\sigma(i)}$ is ever realized from $X_{\sigma(i)}$ in this ordering, then the variable would be rejected. Define $X_{\sigma(i)}'$ to have support $\{0, m_{\sigma(i)}, 1\}$ and probabilities $\{1-p_{\sigma(i)}-q_{\sigma(i)}, p_{\sigma(i)}, q_{\sigma(i)}\}$ for $i=1,\ldots n-1$, and $X_{\sigma(n)}':=E[X_{\sigma(n)}]$. Then, $V(X_{\sigma(1)},\ldots, X_{\sigma(n)}) = V(X_{\sigma(1)}',\ldots, X_{\sigma(n)}')$.

Now, if we knew $\sigma(n)$, then we could transform $X_{\sigma(i)}$ into $X_{\sigma(i)}'$ for $i=1,\ldots,n$ and then, since $X_{\sigma(i)}'$ has support on $\{0,m_i,1\}$, we can use Algorithm~\ref{alg:fptas} and Theorem~\ref{th:FPTAS-same}. Since we do not know $\sigma(n)$, we run our FPTAS in Algorithm~\ref{alg:fptas} $n$ times. Here, in the $i^{th}$ iteration, we define $X_i':=E[X_i]$ and $X_j'$ to have support $\{0, m_j, 1\}$ and probabilities $\{1-p_j-q_j,p_j, q_j\}$ for $j\neq i$. \blue{And the algorithm is run on $X'_j$ variables instead of $X_j$s}. \blue{In one of these iterations (specifically the iteration where $X_i':=E[X_i]$, for $i=\sigma(n)$) 
 the ordering found by the algorithm will satisfy the required guarantees}.

	\section{Conclusions and further directions}\label{Conclusion}
\blue{In this paper, we took significant steps towards a comprehensive understanding of the optimal ordering problem when the distributions involved have support on a constant number of points. We provided a very strong hardness result that shows the problem is NP-hard even for a very special case of $3$ point distributions. Subsequently, we closed  the problem for $2$-point distributions, as well as the said special case of $3$-point distributions, by providing a polynomial time algorithm and an FPTAS respectively. We also provided insights on the impact of ordering by proving improved prophet inequalities.}

There is much left to investigate.
An open question is whether the FPTAS derived in Section \ref{sec:fptas} can be extended to $k$-point distributions for any constant $k$ (we know from \cite{fu2018ptas} that a PTAS is possible). \blue{Our hardness result does not  rule out the possibility of such an algorithm. }\sacomment{...although it makes it unlikely..? should we make comments about strong NP-hardness here? }
We proved that for two-point distributions, the expected reward under optimal ordering is within a factor of $1.25$ of the prophet's reward, thus improving the well-known prophet inequality for worst-case ordering (from factor $2$ to $1.25$). Can such a prophet inequality be proven for best ordering in general $k$-point distributions? Finally, an interesting direction is to conduct such an investigation into optimal ordering for other parametric forms of distributions. 
	\bibliographystyle{ACM-Reference-Format}
	\bibliography{bib}
	\appendix
	\newpage

\section{Missing proofs from Section \ref{sec:hardness}}
\label{app:hardness}

\begin{claim}
Suppose $X_1, X_2, \ldots, X_n$ is an optimal ordering of the random variables. In an optimal stopping rule for this ordering, let $S$ be the set of random variables that are accepted only when their realization is $1$, and $T$ be the random variables that are accepted whenever their realization is positive. Then, $S$ precedes $T$. That is, $i < j$ whenever $X_i \in S$ and $X_j \in T$.
\end{claim}
\begin{proof} Suppose there is an optimal ordering for which $S$ does not precede $T$. In such a case there must be a pair of {\em adjacent} random variables $X_i, X_j$ in the ordering such that $X_i \in S$, $X_j \in T$, and $X_j$ appears before $X_i$. Let $L$ be the sequence of random variables that precede $X_j$ and $R$ be the sequence of random variables that succeed $X_i$. 
We prove the claim by a standard interchange argument in which $X_i$ and $X_j$ are swapped: not surprisingly, the contributions from $L$ and from $R$ will be the same in both sequences, so their difference will assume a simple form.

Let $\sigma_{ij}$ be the ordering where $X_i$ precedes $X_j$ and let $\sigma_{ji}$ be the interchanged ordering where $X_j$ precedes $X_i$. For their fixed thresholds, let $f(L)$  be the probability that none of the random variables in $L$  is accepted and let $E(L)$ be the expected reward given that a random variable in $L$ is accepted (similarly define for $R$). Define $V_{\sigma_{ij}}$ and $V_{\sigma_{ji}}$ to be the expected rewards under orderings $\sigma_{ij}$ and $\sigma_{ji}$ respectively, and under these fixed thresholds.
Then, it is easy to verify that
$$V_{\sigma_{ij}} \; = \; E(L)(1-f(L)) + f(L) [ q_i + (1-q_i) (p_j m_j + q_j) + (1-q_i)(1-q_j-p_j) E(R)(1-f(R))]$$
and
$$V_{\sigma_{ji}} \; = \; E(L)(1-f(L)) + f(L) [ p_j m_j + q_j + (1-q_j-p_j) q_i + (1-q_j-p_j) (1-q_i) E(R)(1-f(R))].$$
Simplifying, we have:
$$V_{\sigma_{ij}} - V_{\sigma_{ji}} \;\; = \;\; f(L) p_j q_i (1 - m_j) \; > 0.$$
Thus, swapping $i$ and $j$ while retaining their acceptance thresholds improves the original ordering, which, therefore, cannot be optimal.
\end{proof}

\begin{claim}
Suppose $(S,T)=(S^\sigma, T^\sigma)$ for some ordering $\sigma$. Then, the ordering $\sigma$ is optimal if and only if (i) $S$ precedes $T$; (ii) the random variables in $S$ are arranged arbitrarily; and (iii) the random variables in $T$ appear in weakly decreasing order of $E_i$.
In particular, if $E_1 = E_2 = \ldots = E_n$, then the random variables in $T$ can be arranged arbitrarily as well. 
\end{claim}

\noindent
{\bf Proof.}
We already know that any ordering in which $S$ does not precede $T$ is sub-optimal, verifying (i). To see (ii), note that any random variable in $S$ that is accepted results in a value of 1, and that the probability of accepting {\em some} random variable in $S$ is $1 - \Pi_{i: X_i \in S} (1-q_i)$, regardless of how these random variables are ordered. We can verify (iii) using a simple interchange argument as well. 
Suppose $X_i, X_j \in T$. As before, we let $\sigma_{ij}$ be an ordering in which $X_i$ appears immediately before $X_j$, with
$L$ being the sequence of random variables that precede $X_i$ and $R$ being the sequence of random variables that succeed $X_j$. $\sigma_{ji}$ is the ordering with $X_i$ and $X_j$ interchanged.
Let $f(L)$  be the probability that none of the random variables in $L$  is accepted and let $E(L)$ be the expected reward given that a random variable in $L$ is accepted (similarly define for $R$). Define $V_{\sigma_{ij}}$ and $V_{\sigma_{ji}}$ to be the expected rewards under orderings $\sigma_{ij}$ and $\sigma_{ji}$ respectively, and under these fixed thresholds. 
Then, it is easy to verify that
$$V_{\sigma_{ij}} \; = \; E(L)(1-f(L)) + f(L) [ p_i m_i + q_i + (1-q_i-p_i) (p_j m_j + q_j) + (1-q_i-p_i)(1-q_j-p_j) E(R)(1-f(R))]$$
and
$$V_{\sigma_{ji}} \; = \; E(L)(1-f(L)) + f(L) [ p_j m_j + q_j + (1-q_j-p_j) (p_i m_i + q_i) + (1-q_j-p_j) (1-q_i-p_i) E(R)(1-f(R))].$$
Simplifying, we have:
\begin{eqnarray*}
V_{\sigma_{ij}} - V_{\sigma_{ji}} &  = & f(L) [ (p_j + q_j) (p_i m_i + q_i) - (p_i+q_i) (p_j m_j + q_j) ] \\
& = & \frac{f(L)}{(p_i+q_i)(p_j+q_j)} (E_i - E_j).
\end{eqnarray*}
Thus, it is optimal for $X_i$ to appear before $X_j$ in $T$ if $E_i > E_j$.
To see that any ordering satisfying properties $(i)-(iii)$ must be optimal note that the value of any ordering satisfying all of these properties is identical, and so should be optimal (because $(S,T)$ is assumed to be an optimal partition).
\qed
\section{Missing proofs for section \ref{sec:prophet}}
\label{app:prophet}
We prove the following statement to complete the proof of Theorem \ref{th:prophet} in Section \ref{sec:prophet}. 
\begin{lemma}
$$ \max(T1, T2, T3) \ge 0.8 \text{MAX}$$
where $\text{MAX}$ is as defined in \eqref{eq:max} and $T1, T2, T3$ are as defined in \eqref{eq:T123}.
\end{lemma}

Using some algebraic manipulations, we can equivalently express MAX defined in \eqref{eq:max} as:
\begin{eqnarray*}
MAX& :=& p^*b^* + (1-p^*)p_w b_w + (1-p^*)(1-p_w)a^*\\
&=& p^*b^*+p_w b_w -p^*p_w b_w+p^*p_wb^*-p^*p_wb^* + (1-p^*)(1-p_w)a^*\\
&=& p^*p_w(b^*-b_w) + (1-p_w)p^*b^*+p_w b_w + (1-p^*)(1-p_w)a^*\\
&=& f(p^*, p_w, b^*, b_w) + g(p^*, p_w, b^*, b_w) + h(p^*, p_w, a^*)
\end{eqnarray*}

where 

$$f(p^*, p_w, b^*, b_w):= p^*p_w(b^*-b_w)$$
$$g(p^*, p_w, b^*, b_w) := (1-p_w)p^*b^*+p_w b_w$$
$$h(p^*, p_w, a^*) := (1-p^*)(1-p_w)a^*$$

Now, notice that  is $T1=g + h$ and $T2=f + g$. The idea of the proof is to bound the relative fraction of $f$ or $h$ to the offline expectation. We assume w.l.o.g. that $b_w=1$. We can do this by multiplying every random variable by some appropriate constant $\alpha$, such that $b_w = 1$. This scales the prophet's reward by $\alpha$ since $E[\max\{\alpha X, \alpha Y\}] = \alpha E[\max\{X, Y\}]$. Furthermore, $V(\alpha X, \alpha Y) = E[\max\{\alpha X, V(\alpha Y)\}]=\alpha E[\max\{X, V(Y)] = \alpha V(X, Y)$. Thus, the optimal reward also scales by $\alpha$ and so the competitive ratio remains the same.

The following three claims together prove the lemma statement.
\begin{claim}\label{Policy_1_best}
If $T1\ge \max(T2, T3)$, then $T1\ge 0.8 MAX$
\end{claim}

\begin{proof}
	
	We can assume that $\mu^* = p^*b^*+(1-p^*)a^* \leq 1$ otherwise either $T2$ or $T3$ would be greater than $T_1$. Now since $T1\geq T2$, this implies that $f\leq h$, and combined with the previous inequality, we get 
	
	\begin{eqnarray*}
	p^*p_w(b^*-1) &\leq & (1-p^*)(1-p_w)a^*\\
	&\leq& (1-p_w)(1-p^*b^*)
	\end{eqnarray*}
	
	Rearranging terms, we get $p_w\leq \frac{1-p^*b^*}{1-p^*}$. Next, we will prove the following:
	
	$$f(p^*, p_w, b^*) <  \frac{1}{3}g(p^*, p_w, b^*)$$
	
	The derivation is as follows:
	
	\begin{eqnarray*}
	\frac{g(p^*, p_w, b^*)}{f(p^*, p_w, b^*)}&=& \frac{(1-p_w)p^*b^*+p_w}{p^*p_w(b^*-1)}\\
	&=&\frac{(\frac{1}{p_w}-1)p^*b^*+1}{p^*(b^*-1)}\\
	&\geq & \frac{(\frac{p^*(b^*-1)}{1-p^*b^*})p^*b^*+1}{p^*(b^*-1)}\\
	&=&\frac{(p^*)^2(b^*-1)b^* + 1 -p^*b^*}{p^*(b^*-1)(1-p^*b^*)}\\
	&=&\frac{p^*b^*}{1-p^*b^*} + \frac{1}{p^*(b^*-1)}\\
	&=&\frac{t}{1-t} + \frac{1}{t-p^*}
	\end{eqnarray*}
	
	Where we let $t=p^*b^*$. For fixed $p^*$, this term is minimized when $t = \frac{1}{2}(1+p^*)$. Plugging this in, the minimum is $\frac{3+p^*}{1-p^*}$. This is minimized when $p^*\to 0$. Thus, $\frac{f(p^*,p_w, b^*)}{g(p^*, p_w, b^*)} < \frac{1}{3}$.\\
	
	When $T1$ is maximum, it gives a competitive ratio of $\frac{g+h}{f+g+h}\geq \frac{g+f}{g+2f}\geq \frac{(4/3)g}{(5/3)g}=4/5$
\end{proof}

\begin{claim}\label{Policy_2_best}
If $T2\ge \max(T1, T3)$, then $T2\ge 0.8 MAX$
\end{claim}

\begin{proof}
	
	We will prove that 
	
	$$h(p^*, p_w, a^*)\leq \frac{1}{3}g(p^*, p_w, b^*)$$
	
	Notice that since $T2\geq T3$, we must have $a^*\leq p_w$ and $b^*\geq 1$. Furthermore, $h\leq f$ so combining these inequalities, we have the relation $(1-p^*)(1-p_w)a^*\leq (1-p^*)(1-p_w)p_w\leq p_wp^*(b^*-1)\Rightarrow b^*\geq \frac{(1-p^*)(1-p_w)}{p^*}+1$. Using these two inequalities, 
	
	\begin{eqnarray*}
	\frac{g(p^*, p_w, b^*)}{h(p^*, p_w, a^*)} &=& \frac{(1-p_w)p^*b^*+p_w}{(1-p^*)(1-p_w)a^*}\\
	&\geq& \frac{(1-p_w)^2(1-p^*)+(1-p_w)p^*+p_w}{(1-p^*)(1-p_w)p_w}\\
	&=& \frac{1-(1-p^*)(1-p_w)p^*}{(1-p^*)(1-p_w)p_w}\\
	&=&\frac{1}{(1-p^*)(1-p_w)p_w} - \frac{p^*}{p_w}
	\end{eqnarray*}
	
	Differentiating this with respect to $p^*$, we see that the minimum is when $p^* = 1-p_w$. Putting this back in, we have to minimize the term 
	
	$$\frac{1}{(1-p_w)^2p_w} -\frac{1-p_w}{p_w}$$
	
	The minimum is $3$, attained at $p_w = 0$, proving that $h(p^*, p_w, a)\leq \frac{1}{3}g(p^*, p_w, b^*)$.\\
	
	When $T2$ is maximum, it gives a competitive ratio of $\frac{f+g}{f+g+h} \geq \frac{g+h}{g+2h}\geq\frac{(4/3)g}{(5/3)g}=4/5$.
\end{proof}

\begin{claim}\label{Policy_3_best}
	If $T3 \ge \max(T1, T2) $, then $T3 \ge 0.8 MAX$. 
\end{claim}
\begin{proof}
	
	$T3$ has an expected reward of $p^*b^*+(1-p^*)a^*$. Since $T3\geq T2$, we have $a^*\geq p_w$, and since $T3\geq T1$, we have $p^*b^* + (1-p^*)a^* \geq p_w + (1-p_w)p^*b^* + (1-p_w)(1-p^*)a^*$. Now suppose we swap the variables $a^*$ and $p_w$ in the expected reward and in the constraints. Then the ``expected reward'' is 
	
	$$p^*b^* + (1-p^*)p_w$$
	
	and we have the constraints
	$$p_w \leq a^*$$
	$$p^*b^* + (1-p^*)p_w\geq a^* + (1-a^*)p^*b^*+(1-a^*)(1-p^*)p_w$$
	
	The second constraint can be rearranged to show that $b^*\geq \frac{(1-p^*)(1-p_w)}{p^*} + 1$. Thus, this case can be reduced to the case of Claim \ref{Policy_2_best}, so the rest of the proof is the same.
\end{proof}


\section{Analysis of Algorithm \ref{alg:fptas}: Proof of Theorem \ref{th:FPTAS-same}}
We show that Algorithm \ref{alg:fptas} is an FPTAS for the problem in \eqref{eq:paritionProblem} of finding an optimal ordered partition.  Algorithm \ref{alg:fptas} discretizes the interval $[0,1]$ using a multiplicative grid with parameter $1-\frac{\epsilon}{2n}$,  so that it needs to search over only $poly(n,1/\epsilon)$ partitions.  Lemma \ref{Partition-Subset} allows us to restrict to these grid points.
Then, in Lemma \ref{ALG-OPT approx} and Lemma \ref{lem:runtime}, respectively, we show that Algorithm \ref{alg:fptas} achieves the required approximation and run-time, to complete the proof of Theorem \ref{th:FPTAS-same}.
\begin{lemma}\label{Partition-Subset}
 Let $\mathcal{L}' \subseteq \mathcal{L}$ be the set of all ordered partitions $\{(S',T')\}$ in $\mathcal{L}$ with additional restriction on the last variable $X$ in $(S,T)$ that $V(X) = E[\max(X,0)] \ge \frac{\epsilon \OPT}{2n}$ for $0 \le \epsilon \le 1$.  Let $\OPT'=\max_{(S',T')\in \mathcal{L}'} V(S',T')$. Then, 
 $$\OPT'\ge \left(1-\frac{\epsilon}{2}\right) \OPT.$$
\end{lemma}
\begin{proof}
	Let $\delta := \frac{\epsilon \OPT}{2n}$. W.l.o.g., assume  that an optimal ordered partition $(S,T)$ orders the variables as $(X_1,\ldots, X_n)$, so that $\OPT=V(X_1, \ldots, X_n)$. Suppose that there exist $1\le k\le n$ such that $V(X_k)> \delta$. Then,
	\begin{eqnarray*}
	\OPT'\ge  V(X_1, \ldots, X_k) & \ge & V(X_1,\ldots, X_n)-\sum_{i=k+1}^n V(X_i) \\
	& \ge & \OPT - (n-1)\delta\\
	& \ge & (1-\frac{\epsilon}{2}) \OPT,
	\end{eqnarray*}
	where the second inequality followed from repeatedly applying Lemma $\ref{Adding to Tail}$. 
	Now if such a $k$ does not exist, then by Lemma \ref{Adding to Tail}, 
	$$ \OPT = V(X_1,\ldots, X_n) \le \sum_{i=1}^n V(X_i) \le n\delta \le \frac{\epsilon}{2} \OPT$$
	so that trivially, $ (1-\frac{\epsilon}{2})\OPT\le 0 \le\OPT'$.

\end{proof}

\newcommand{\threshold}{H}
\begin{lemma}\label{ALG-OPT approx}
Let  $\mathcal{L}^n$ be the set of ordered partitions returned by Algorithm \ref{alg:fptas} when run with  parameters $\epsilon, \MAX$ satisfying $\epsilon \in (0,1)$ and 
$\MAX\le \OPT$. And, let
$\ALG:=\max_{(S,T)\in \mathcal{L}^n} V(S,T), $   Then, 
	$$\ALG \geq (1-\frac{\epsilon}{2n})^n \OPT', $$ where $\OPT'$ is as defined in Lemma \ref{Partition-Subset}.
\end{lemma}
\begin{proof}
     Let $(S,T)$ be any ordered partition  in $ \mathcal{L}'$ where $\mathcal{L}' \subseteq \mathcal{L}$ is the restricted collection of ordered partitions defined in Lemma \ref{Partition-Subset} satisfying $V(X) = E[\max(X,0)] \ge \frac{\epsilon \OPT}{2n}$ for the last variable $X$ in $T$ (note that $T\ne \phi$ for all $(S,T)\in \mathcal{L}$). We show that there exist $(S', T')\in \mathcal{L}^n$ such that $V(S', T') \ge (1-\frac{\epsilon}{2n})^n V(S,T)$. 
    
    We prove this by induction. Let $(S_k, T_k)$ be an ordered partition obtained on restricting $S, T$ to the first $k$ variables $X_1, \ldots, X_k$ considered by the algorithm (here variables are ordered so that $E_1 \le \cdots \le E_k$). Let $1\le \bar{k}\le n$ be such that $X_{\bar{k}}$ is the last variable in $T$ and $V(X_{\bar k}) \ge \frac{\epsilon \OPT}{2n}$; such a $\bar{k}$ must exist since  $(S,T)\in \mathcal{L}'$ and $T\ne \phi$.
    We show that for all $k\ge \bar{k}$, at the end of the iteration $k$ of the algorithm, there exists $(S'_k, T'_k)\in \mathcal{L}^k$ such that 
    \begin{equation}
    \label{eq:induction}
    \begin{array}{rcl}
    V(S'_k)  & \ge &  V(S_k),\\ 
     V(T'_k)  & \ge & \rho^k V(T_k), \\
       V(T'_k) & \ge & \frac{\epsilon}{2n}\MAX,
     \end{array}
     \end{equation}
    where $\rho=(1-\frac{\epsilon}{2n})$.
    
    We prove the induction basis for $k=\bar{k}$.  By definition of $\bar{k}$, $S_k=\{X_1, \ldots, X_{k-1}\}$ and $T_k=\{X_k\}$. 
    Since $(\{X_1,\ldots, X_{k-1}\}, \phi)\in \mathcal{L}^{k-1}$, in the beginning of iteration $k$ (before TRIM),  the partition  $(S''_k, T''_k) = (\{X_1, \ldots, X_{k-1}\}, \{X_k\})$ is added to $\mathcal{L}^k$. Since 
    \begin{center}
        $V(T''_k) = V(X_k)\ge \frac{\epsilon}{2n}\OPT \ge \frac{\epsilon}{2n}\MAX,$
    \end{center} during TRIM this partition will fall in bucket $\mathcal{B}_j$ for some $j\ge 1$. By the trimming criteria, one partition $(S'_{k}, T'_k) $ from bucket $\mathcal{B}_j$ will remain in $\mathcal{L}^k$ satisfying     $V(T'_k) \ge \rho V(T''_k) =\rho V(T_k)\ge \rho^k V(T_k)$,
        $V(T'_k)\ge \rho^J max \ge \frac{\epsilon}{2n}\MAX$, 
        and $V(S'_k)\ge V(S''_k) \ge V(S_k)$.
    Therefore, $S'_k, T'_k$ satisfies the conditions stated in \eqref{eq:induction} for $k=\bar{k}$.
    
    Now, for the induction step, assume \eqref{eq:induction} holds for some $\bar{k} \le k<n$.  Then, in the beginning of iteration $k+1$ (before $\text{TRIM}$ is called), the algorithm will add two partitions $(\{X_{k+1}, S'_k\}, T'_k)$ and $(S'_k, \{X_{k+1}, T'_k\})$ to $\mathcal{L}^{k+1}$. 
    We claim that one of these two partitions satisfies the required induction statement for $k+1$, but with a better factor $\rho^k$ instead of the required $\rho^{k+1}$. This can be observed as follows. 
    
    Depending on whether $X_{k+1}$ appears in $T_{k+1}$ or $S_{k+1}$, we can consider two cases: either $T_{k+1}=T_k$ or $S_{k+1}=S_k$. In the first case (when $T_{k+1}=T_k$), we set $(S'_{k+1}, T'_{k+1})$ as the first partition $(\{X_{k+1}, S'_k\}, T'_k)$. By induction hypothesis $V(T'_{k+1})=V(T'_k) \ge \rho^k V(T_k) = \rho^k V(T_{k+1})$; also $V(T'_{k+1})=V(T'_k) \ge \frac{\epsilon}{2n} \MAX $; and 
    \begin{eqnarray}
    \label{eq:tmp0}
    V(S'_{k+1}) = V(X_{k+1}, S'_k) & =  &  V(X_{k+1}, V( S'_k))\nonumber\\
    & \ge  &  V(X_{k+1}, V(S_k)) \nonumber\\
    & = & V(X_{k+1}, S_k)\nonumber\\
    & = & V(S_{k+1}) 
    \end{eqnarray}
    where the second line follows from the induction hypothesiss.
    
    For the second case (when $S_{k+1}=S_k$), we use the second partition, and set $(S'_{k+1}, T'_{k+1})= (S'_k, \{X_{k+1}, T'_k\})$, so that by induction hypothesis, $V(S'_{k+1})= V(S'_k) \ge V(S_k) = V(S_{k+1})$, and
    \begin{eqnarray}
    \label{eq:tmp1}
       V(T'_{k+1}) = V(X_{k+1}, T'_k) & =  &  V(X_{k+1}, V( T'_k))\nonumber\\
    & \ge  &  V(X_{k+1}, \rho^{k} V(T_k)) \nonumber\\
    & \ge & \rho^k V(X_{k+1}, T_k)\nonumber\\
    & = & \rho^k V(T_{k+1}) 
    \end{eqnarray}
    where the second line follows from the induction hypothesis and the third line follows from Lemma $\ref{lemma-3}$. Also,  $V(T'_{k+1}) = V(X_{k+1}, T'_k) \ge V(T'_k) \ge \frac{\epsilon}{2n} \MAX$ by induction hypothesis.  
    
    

   
   However, one or both of these two partitions may be removed by the $\text{TRIM}$ procedure. We claim that 
   if any of the two partitions $(S'_{k+1}, T'_{k+1})\in \{(\{X_{k+1}, S'_k\}, T'_k),  (S'_k, \{X_{k+1}, T'_k\})\}$ is   removed by the $TRIM$ procedure, then there will remain another partition  $(S''_{k+1}, T''_{k+1})$ in $\mathcal{L}^{k+1}$ satisfying:
   
   \begin{equation}
   \label{eq:tmp2}
   \begin{array}{rcl}
     V(S''_{k+1}) & \ge & V(S'_{k+1}), \\
     V(T''_{k+1}) & \ge & \rho V(T'_{k+1}) ,\\
     V(T''_{k+1}) & \ge & \frac{\epsilon}{2n}\MAX
\end{array}
     \end{equation}
    
    To see \eqref{eq:tmp2}, note that since  $V(T'_{k+1})\ge \frac{\epsilon}{2n} \MAX$, $(S'_{k+1}, T'_{k+1})$ falls in a bucket $\mathcal{B}_j, j\ne 0$  during the $\text{TRIM}$ procedure.  
    Thus, the TRIM procedure will select one partition from this bucket, let it be $(S''_{k+1}, T''_{k+1})$. By definition of buckets, $V(T''_{k+1})\ge \frac{\epsilon}{2n} \MAX$. Also, by the criteria for selecting a partition from a bucket, we have $V(S''_{k+1})\ge  V(S'_{k+1})$, and by construction of buckets, if $j\ne 0$,  $V(T''_{k+1})\ge \rho V(T'_{k+1})$. 

    
    Together, \eqref{eq:tmp0}, \eqref{eq:tmp1}, \eqref{eq:tmp2} prove the induction statement in \eqref{eq:induction}. Applying \eqref{eq:induction} for $k=n$, we get that there exists $(S'_n, T'_n) \in \mathcal{L}^n$ satisfying
    
     \begin{eqnarray*}
  V(S'_n, T'_n) & = &  V(S'_n,  V( T'_n)))\\
     & \ge &   V(S'_n,  \rho^n V(T_n))\\
     & \ge &   V(S_n,  \rho^n V(T_n))\\
    & \ge &   \rho^n V( S_n, V(T_n))\\
    & = & \rho^n V(S,T)
    \end{eqnarray*}
   Here  the first inequality followed from  $V(T_n') \ge \rho^n V(T_n)$. For the second inequality, note that a variable in $S_n$ (and $S_n'$) is accepted  if and only if it takes value $1$. Therefore, $S_n$  can be replaced by a $\{0,1\}$ variable $Y$ with probability $\prod_{i\in S_n} q_i$ to take value $1$ (and similarly $S_n'$ can be replaced by a $\{0,1\}$ variable $Y'$). Then, since we have $E[Y']=V(S'_n) \ge V(S_n) = E[Y]$, the second inequality follows from  Lemma \ref{lemma-2}. The third inequality follows from Lemma \ref{lemma-3}. 
    
    This completes the proof of the lemma.
    \end{proof}

\begin{lemma}
\label{lem:runtime}
Algorithm \ref{alg:fptas} with parameters $\epsilon\in (0,1)$ and $\MAX\ge \frac{\OPT}{2}$  runs in $O(\frac{n^4}{\epsilon^2})$ time.
\end{lemma}

\begin{proof}
 Given $\MAX\ge \frac{\OPT}{2}$, in the TRIM procedure (Algorithm \ref{alg:trim}), we always have $\frac{max}{\MAX} \le \frac{\OPT}{\OPT/2} \le 2$. Therefore,  the condition $\rho^J max \ge \frac{\epsilon}{2n} \MAX$ in the TRIM procedure ensures that the number of buckets 
 
 $$J\le  \log_{1/\rho}(\frac{2n}{\epsilon} \frac{max}{\MAX}) \le \log_{1/\rho}(\frac{4n}{\epsilon})=O(\frac{1}{(1-\rho)}\frac{n}{\epsilon})= O(\frac{n^2}{\epsilon^2})$$
	
	
Therefore, we maintain $O(\frac{n^2}{\epsilon^2})$ partitions in each iteration. Since for each partition, we need to calculate the expected reward, which is $O(n)$ time, and there are $n$ iterations, we get the lemma statement.
\end{proof}
Now, we are ready to prove Theorem \ref{th:FPTAS-same}.
\paragraph{Proof of Theorem \ref{th:FPTAS-same}}
Let $\mathcal{L}^n$ be the set of partitions returned by Algorithm \ref{alg:fptas} with parameters $\epsilon\in (0,1)$, and 
 \begin{center} $\MAX:=\frac{1}{2} E[\max(X_1,\ldots, X_n)].$ \end{center}
 Then, $\MAX\geq \frac{1}{2} \OPT$, so that by Lemma \ref{lem:runtime}, Algorithm \ref{alg:fptas} runs in time $O(\frac{n^4}{\epsilon^2})$ time. 
 Also, using prophet inequality  \cite{samuel1984comparison}, $\MAX \le \OPT$, 
 so that by Lemma \ref{Partition-Subset} and Lemma \ref{ALG-OPT approx}, 
 \begin{center}
     $\ALG\ge (1-\epsilon/2) \OPT' \ge (1-\epsilon/2)^2 \OPT \ge (1-\epsilon)\OPT.$
 \end{center}

\section{Other algebraic  lemmas }

We used following lemmas in the analysis.

\begin{lemma}\label{scaling}
	\textbf{Additive Scaling:} Given random variables $X_1,\ldots X_n$ and $c\in\mathbb{R}$ such that $Y_i:= X_i + c$ is a non-negative random variable. Let $\sigma$ be a permutation. Then $V(Y_{\sigma(1)}, \ldots, Y_{\sigma(n)}) = V(X_{\sigma(1)}, \ldots, X_{\sigma(n)}) + c$.
\end{lemma}
\begin{proof}
	We prove by induction. For one variable, $V(Y) = E[Y] = E[X + c] = E[\max\{X, 0\}] + c = V(X) + c$.
	W.l.o.g, let $\sigma = (1, 2, \ldots, k+1)$. For the inductive step:
	
	\begin{eqnarray*}
	V(Y_1,\ldots, Y_{k+1}) &=& V(Y_1, V(Y_2\ldots, Y_{k+1}))\\
	&=&V(X_1 + c, V(X_2\ldots, X_{k+1}) + c)\\
	&=&E[\max(X_1, V(X_{2}, \ldots, X_{k+1})] + c\\
	&=&	V(X_1,\ldots, X_{k+1}) + c
	\end{eqnarray*}

    where the second line follows from the induction hypothesis.
\end{proof}
\begin{lemma}\label{Adding to Tail}
	For any $v \ge 0$, $E[\max\{X, c+v\}]\leq E[\max\{X, c\}]+v$ and $E[\max\{X,c-v\}]\geq E[\max\{X,c\}]-v$
\end{lemma}


\begin{lemma}\label{lemma-2}
	Let $Y_1$ and $Y_2$ be two $\{0,1\}$ random variables where $E[Y_1] \geq E[Y_2]$. Then $E[\max\{Y_1, c\}] \geq E[\max\{Y_2,c\}]$ for any constant $0\leq c < 1$.
\end{lemma}


\begin{lemma}\label{lemma-3} 
	$E[\max\{X, \delta c\}]/E[\max\{X, c\}]\geq \delta$ for $0\leq \delta\leq 1$
\end{lemma}
\begin{proof}
		For convenience we denote $V(X,c):=E[\max\{X, c\}]$
	
	\begin{eqnarray*}
 	\frac{V(X,\delta c)}{V(X,c)}&\geq& \frac{V(X,c)-(c-\delta c)}{V(X,c)}\\
 	&=& 1-\frac{c(1-\delta)}{V(X,c)}\\
 	&\geq& 1-(1-\delta)\\
 	&=&\delta 
 	\end{eqnarray*}
	
 	where in the first line, we used Lemma $\ref{Adding to Tail}$ and in the third line, we used $V(X,c)\geq c$
 \end{proof}

 \newpage
\end{document}